\documentclass[11pt]{article}

\usepackage{amsmath}
\usepackage{amsthm}
\usepackage{amssymb}
\usepackage{amsfonts}
\usepackage{graphicx}
\usepackage{epsfig}
\usepackage{epsf}
\usepackage{datetime}
\usepackage{latexsym}
\usepackage{enumerate}
\usepackage[small]{caption}
\usepackage{amsfonts}
\usepackage{fancyvrb}
\usepackage{fancyhdr}
\usepackage{amssymb}
\usepackage{amsthm}
\usepackage{graphics}
\usepackage{psfrag}
\usepackage{bbm}
\usepackage{url}
\usepackage{hyperref}

\newtheorem{theorem}{Theorem}
\newtheorem{lemma}{Lemma}

\newtheorem{corollary}{Corollary}

\newtheorem{observation}{Observation}

\newtheorem{problem}{Problem}

\def\etal{{et~al.}}

\def\ie{{i.e.}}

\newcommand{\alg}{\textsf{ALG}}
\newcommand{\opt}{\textsf{OPT}}

\newcommand{\NN}{\mathbb{N}} 
\newcommand{\ZZ}{\mathbb{Z}} 
\newcommand{\RR}{\mathbb{R}} 
\newcommand{\eps}{\varepsilon}
\providecommand{\intd}[0]%
{\;\mbox{d}}

\newcommand{\old}[1]{{}}
\newcommand{\later}[1]{{}}

\renewcommand{\vec}[1]{\overrightarrow{#1}}

\title{\textsc{Online Unit Covering in Euclidean Space}\footnote{
Preliminary results were presented at the 27th Annual Fall Workshop on Computational
Geometry, Nov.~2017, University of New York at Stony Brook, USA.}
}

\author{Adrian Dumitrescu\thanks{Department of Computer Science,
University of Wisconsin--Milwaukee, WI, USA\@.
Email:~\texttt{dumitres@uwm.edu}.}
 \qquad
 Anirban Ghosh\thanks{School of Computing,
   University of North Florida,
   Jacksonville, FL, USA\@.
   Email:~\texttt{anirban.ghosh@unf.edu}.
Research supported by the University of North Florida start-up fund.}
  \qquad
 Csaba D. T\'oth\thanks{Department of Mathematics,
    California State University Northridge, Los Angeles, CA;
 and Department of Computer Science, Tufts University, Medford, MA, USA\@.
 Email:~\texttt{cdtoth@acm.org}.}
}

\begin{document}
\maketitle

\begin{abstract}
We revisit the online \textsc{Unit Covering} problem in higher dimensions:
Given a set of $n$ points in $\RR^d$, that arrive one by one,
cover the points by balls of unit radius,
so as to minimize the number of balls used.
In this paper, we work in $\RR^d$ using Euclidean distance.
The current best competitive ratio of an online algorithm, $O(2^d d \log{d})$,
is due to Charikar~\etal~(2004); their algorithm is deterministic.

(I) We give an online deterministic algorithm with competitive ratio $O(1.321^d)$,
thereby sharply improving on the earlier record by a large exponential factor.
In particular, the competitive ratios are $5$ for the plane and $12$ for $3$-space
(the previous ratios were $7$ and $21$, respectively).
For $d=3$, the ratio of our online algorithm matches the ratio
of the current best offline algorithm for the same problem due to Biniaz~\etal~(2017),
which is remarkable (and rather unusual).

(II) We show that the competitive ratio of every deterministic online algorithm
(with an adaptive deterministic adversary) for \textsc{Unit Covering}
in $\RR^d$ under the $L_{2}$ norm is at least $d+1$ for every $d \geq 1$.
This greatly improves upon the previous best lower bound, $\Omega(\log{d} / \log{\log{\log{d}}})$,
due to Charikar~\etal~(2004).

(III) We obtain lower bounds of $4$ and $5$ for the competitive ratio of any deterministic algorithm
for online \textsc{Unit Covering} in $\RR^2$ and respectively $\RR^3$;
the previous best lower bounds were both $3$.

(IV) When the input points are taken from the square or hexagonal lattices in $\RR^2$,
we give deterministic online algorithms for \textsc{Unit Covering} with an optimal
competitive ratio of $3$.


\end{abstract}

\section{Introduction} \label{sec:intro}

Covering and clustering are fundamental problems in the theory of algorithms,
computational geometry, optimization, and other areas.
They arise in a wide range of applications, such as facility
location, information retrieval, robotics, and wireless networks.
While these problems have been studied in an offline setting for
decades, they have been considered only recently in a more
dynamic (and thereby practical) setting. Here we study such problems in a
high-dimensional Euclidean space and mostly in the $L_2$ norm.
We first formulate them in the classic \emph{offline} setting.

\begin{problem} \label{prob:1}
\textsc{$k$-Center}. Given a set of $n$ points in $\RR^d$ and a positive integer $k$,
cover the set by $k$ congruent balls centered at the points so that the diameter
of the balls is minimized.
\end{problem}

The following two problems are dual to Problem~\ref{prob:1}.

\begin{problem}\label{prob:2}
\textsc{Unit Covering}. Given a set of $n$ points in $\RR^d$,
cover the set by balls of unit diameter so that the number of balls is minimized.
\end{problem}

\begin{problem}\label{prob:3}
\textsc{Unit Clustering}.
Given a set of $n$ points in $\RR^d$,
partition the set into clusters of diameter at most one so that the number of clusters is minimized.
\end{problem}

Problems~\ref{prob:1} and~\ref{prob:2} are easily solved in polynomial time for points on the line,
\ie, for $d=1$; but both problems become NP-hard already in
Euclidean plane~\cite{FPT81,MS84}. Factor $2$ approximations are known for
\textsc{$k$-Center} in any metric space (and so for any dimension)~\cite{FG88,Go85};
see also~\cite[Ch.~5]{Va01},~\cite[Ch.~2]{WS11},
while polynomial-time approximation schemes are known for \textsc{Unit Covering}
for any fixed dimension~\cite{HM85}. However, these algorithms are notoriously inefficient
and thereby impractical; see also~\cite{BLMS17} for a summary of such results and different
time vs. ratio trade-offs.

Problems~\ref{prob:2} and~\ref{prob:3} are identical in the offline setting:
indeed, one can go from clusters to balls in a straightforward way; and conversely,
one can assign multiply covered points in an arbitrary fashion to unique balls.
In regard to their \emph{online} versions, it is worth emphasizing two common properties:
(i)~a point assigned to a cluster must remain in that cluster; and
(ii)~two distinct clusters cannot merge into one cluster, \ie, the clusters maintain
their identities. In this paper we focus on the second problem, namely \emph{online}
\textsc{Unit Covering}; we however point out key differences between this problem and
online \textsc{Unit Clustering}.

The performance of an online algorithm $\alg$ is measured by comparing it to an
optimal offline algorithm $\opt$ using the standard notion of competitive ratio~\cite[Ch.~1]{BY98}.
The competitive ratio of $\alg$ is defined as
$\sup_\sigma \frac{\alg(\sigma)}{\opt(\sigma)}$,
where $\sigma$ is an input sequence of points,
$\opt(\sigma)$ is the cost of an optimal offline algorithm for $\sigma$
and $\alg(\sigma)$ denotes the cost of the solution produced by $\alg$ for this input.
For randomized algorithms, $\alg(\sigma)$ is  replaced by the expectation $E[\alg(\sigma)]$,
and the competitive ratio of $\alg$ is $\sup_\sigma \frac{E[\alg(\sigma)]}{\opt(\sigma)}$.
Whenever there is no danger of confusion, we use $\alg$ to refer to an algorithm
or the cost of its solution, as needed.

Charikar~\etal~\cite[Sec.~6]{CCFM04} studied the online version of \textsc{Unit Covering}
(under the name of ``Dual Clustering'').
The points arrive one by one and each point needs to be assigned to a new or to an
existing unit ball upon arrival; the $L_2$ norm is used in $\RR^d$, $d\in \NN$.
The location of each new ball is fixed as soon as it is opened.
The authors provided a deterministic algorithm of competitive ratio $O(2^d d \log{d})$
and gave a lower bound of $\Omega(\log{d} / \log{\log{\log{d}}})$
on the competitive ratio of any deterministic algorithm for this problem.
For $d=1$ a tight bound of $2$ is folklore; for $d=2$ the best known
upper and lower bounds on the competitive ratio are $7$ and $3$, respectively,
as implied by the results in~\cite{CCFM04}\footnote{Charikar~\etal~\cite{CCFM04}
  claim (on p.~1435) that a lower bound of $4$ for $d=2$ under the $L_2$ norm
  follows from their Theorem~6.2;  but this claim appears unjustified;
  only a lower bound of $3$ is implied. Unfortunately, this misinformation has been
  carried over also by~\cite{CZ09} and~\cite{DT17}.}.

The online  \textsc{Unit Clustering} problem was introduced by
Chan and Zarrabi-Zadeh~\cite{CZ09} in 2006.
While the input and the objective of this problem are identical to those
for \textsc{Unit Covering}, this latter problem is more flexible in that
the algorithm is not required to produce unit balls at any time, but rather
the smallest enclosing ball of each cluster should have diameter \emph{at most} $1$;
furthermore, a ball may change (grow or shift) in time.
The authors showed that several standard approaches for \textsc{Unit Clustering},
namely the deterministic algorithms
\texttt{Centered},
\texttt{Grid}, and
\texttt{Greedy},
all have competitive ratio at most $2$ for points on the line ($d=1$).
Moreover, the first two algorithms are applicable for \textsc{Unit Covering},
with a competitive ratio at most $2$ for $d=1$, as well.
These algorithms naturally extend to any higher dimension (including \texttt{Grid}
provided the $L_\infty$ norm is used).

\begin{quote}
\texttt{Algorithm Centered.}
For each new point $p$, if $p$ is covered by an existing unit ball, do nothing;
otherwise open a new unit ball centered at $p$.
\end{quote}

\begin{quote}
\texttt{Algorithm Grid.} Build a uniform grid in $\RR^d$ where cells are
unit cubes of the form \linebreak
$\prod \, [i_j,i_j+1)$, where $i_j \in \ZZ$ for $j=1,\ldots,d$.
For each new point $p$, if the grid cell containing $p$ is nonempty,
put $p$ in the corresponding cluster; otherwise open a new cluster for the grid cell
and put $p$ in it.
\end{quote}

Since in $\RR^d$ each cluster of $\opt$ can be split into at most $2^d$ grid-cell clusters
created by the algorithm, its competitive ratio is at most $2^d$, and this analysis is tight
for the $L_\infty$ norm. It is worth noting that there is no direct analogue
of this algorithm under the $L_2$ norm.

\medskip
Some (easy) remarks are in order. Any lower bound on the competitive ratio of an online algorithm
for \textsc{Unit Clustering} applies to the  competitive ratio of the same type of algorithm
for \textsc{Unit Covering}. Conversely,
any upper bound on the competitive ratio of an online algorithm for \textsc{Unit Covering}
yields an upper bound on the competitive ratio of the same type of algorithm
for \textsc{Unit Clustering}.

\paragraph{Related work.}
\textsc{Unit Covering} is a variant of \textsc{Set Cover}. Alon~\etal~\cite{AAA+09}
gave a deterministic online algorithm of competitive ratio $O(\log{m} \log{n})$
for this problem, where $n$ is the size of the ground set and $m$ is the number of sets
in the family. Buchbinder and Naor~\cite{BN09b} obtained sharper results under the assumption
that every element appears in at most $\Delta$ sets.

Chan and Zarrabi-Zadeh~\cite{CZ09} showed that no online algorithm
(deterministic or randomized) for \textsc{Unit Covering}
can have a competitive ratio better than $2$ in one dimension ($d=1$).
They also showed that it is possible to get better results
for \textsc{Unit Clustering} than for \textsc{Unit Covering}.
Specifically, they developed the first algorithm with competitive ratio below $2$ for $d=1$, namely
a randomized algorithm with competitive ratio $15/8$.
This fact has been confirmed by subsequent algorithms designed for this problem;
the current best ratio $5/3$, for $d=1$, is due to Ehmsen and Larsen~\cite{EL13},
and this gives a ratio of $2^d \cdot \frac56$ for every $d \geq 2$ (the $L_\infty$ norm is used);
their algorithm is deterministic. The appropriate ``lifting'' technique has been layed out
in~\cite{CZ09,ZC09}.  From the other direction, the lower bound for deterministic algorithms
has evolved from $3/2$ in~\cite{CZ09} to $8/5$ in~\cite{ES10}, and then to $13/8$ in~\cite{KK15}.

Answering a question of Epstein and van Stee~\cite{ES10},
Dumitrescu and T\'oth~\cite{DT17} showed that the competitive ratio of any algorithm
(deterministic or randomized) for \textsc{Unit Clustering} in $\RR^d$
under the $L_{\infty}$ norm must depend on the dimension~$d$;
in particular, it is $\Omega(d)$ for every $d\geq 2$.

Liao and Hu~\cite{LH10} gave a PTAS for a related disk cover problem
(another variant of \textsc{Set Cover}): given a set of $m$ disks of arbitrary radii
and a set $P$ of $n$ points in $\RR^2$, find a minimum-size subset of disks that jointly
cover~$P$; see also~\cite[Corollary~1.1]{MR10}.

\paragraph{Our results.}
(i) We show that the competitive ratio of \texttt{Algorithm Centered}
for online \textsc{Unit Covering} in $\RR^d$, $d\in \NN$, under the $L_{2}$ norm
is bounded by the Newton number of the Euclidean ball in the same dimension.
In particular, it follows that this ratio is $O(1.321^d)$
(Theorem~\ref{thm:centered} in Section~\ref{sec:centered}).
This greatly improves on the ratio of the previous best algorithm due to Charikar~\etal~\cite{CCFM04}.
The competitive ratio of their algorithm is at most $f(d)=O(2^d d \log{d})$, where $f(d)$ is the
number of unit balls needed to cover a ball of radius $2$ (i.e., the doubling constant).
By a volume argument, $f(d)$ is at least $2^d$. In particular $f(2)=7$ and $f(3)=21$~\cite{EdWynn};
see also~\cite{BLMS17}. The competitive ratios of our algorithm are $5$ in the plane and
$12$ in $3$-space, improving the earlier ratios of $7$ and $21$, respectively.

\smallskip
(ii) We show that the competitive ratio of every deterministic online algorithm
(with an adaptive deterministic adversary) for \textsc{Unit Covering}
in $\RR^d$ under the $L_{2}$ norm is at least $d+1$ for every $d \geq 1$
(Theorem~\ref{thm:lower-d} in Section~\ref{sec:lower}).
This greatly improves the previous best lower bound, $\Omega(\log{d} / \log{\log{\log{d}}})$,
due to Charikar~\etal~\cite{CCFM04}.

\smallskip
(iii) We obtain lower bounds of $4$ and $5$ for the competitive ratio of any deterministic algorithm
(with an adaptive deterministic adversary) for \textsc{Unit Covering} in $\RR^2$ and respectively
$\RR^3$ (Theorems~\ref{thm:lower-2} and~\ref{thm:lower-d} in Section~\ref{sec:lower}).
The previous best lower bounds were both $3$.

\smallskip
(iv) For input point sequences that are subsets of the infinite square or hexagonal lattices,
we give deterministic online algorithms for \textsc{Unit Covering} with an optimal
competitive ratio of $3$ (Theorems~\ref{thm:lattice1} and~\ref{thm:lattice2}
in Section~\ref{sec:lattice}).

\paragraph{Notation and terminology.}
For two points $p,q \in \RR^d$, let $d(p,q)$ denote the Euclidean
distance between them. Throughout this paper the $L_2$-norm is used.
The closed ball of radius $r$ in $\RR^d$ centered at point $z=(z_1,\ldots,z_d)$ is
$$ B_d(z,r) = \{ x \in \RR^d \ | \ d(z,x) \leq r \}=
\{(x_1,\ldots,x_d) \ | \ \sum_{i=1}^d (x_i -z_i)^2 \leq r^2 \}. $$
A \emph{unit ball} is a ball of unit radius in $\RR^d$.
The \textsc{Unit Covering} problem is to cover a set of points in $\RR^d$ by a minimum
number of unit balls.

The \emph{unit sphere} is the surface of the $d$-dimensional unit ball centered at
the origin $\mathbf{0}$, namely, the set of points $S_d \subset B_d(\mathbf{0},1)$
for which equality holds:  $\sum_{i=1}^d x_i^2=1$.
A \emph{spherical cap} $C(\alpha)$ of angular radius $\alpha \leq \pi$ and center $P$ on $S_d$
is the set of points $Q$ of $S_d$ for which $\angle{P\mathbf{0}Q} \leq \alpha$; see~\cite{Ra55}.

\section{Analysis of \texttt{Algorithm Centered} for online unit covering in  Euclidean $d$-space}
\label{sec:centered}

For a convex body $C \subset \RR^d$, the \emph{Newton number}
(a.k.a. \emph{kissing number}) of $C$ is the
maximum number of nonoverlapping congruent copies of $C$ that can be arranged
around $C$ so that each of them is touching $C$~\cite[Sec.~2.4]{BMP05}.
Some values $N(B_d)$, where $B_d =B_d(\mathbf{0},1)$, are known exactly for small $d$,
while for most dimensions $d$ we only have estimates.
For instance, it is easy to see that $N(B_2)=6$, and it is known that $N(B_3)=12$
and $N(B_4)=24$.
The problem of estimating $N(B_d)$ in higher dimensions is closely related to the problem of
determining the densest sphere packing and the knowledge in this area is largely incomplete
with large gaps between lower and upper bounds; see~\cite[Sec.~2.4]{BMP05}
and the references therein; in particular, many upper and lower estimates up
to $d=128$ are given in~\cite{BDM15} and~\cite{ERS98}.
In this section, we prove the following theorem.

\begin{theorem} \label{thm:centered}
Let $\varrho(d)$ be the competitive ratio of \texttt{Algorithm Centered} in $\RR^d$
(when using the $L_2$ norm). Then $\varrho(2)  = N(B_2)-1=5$, $\varrho(3)  = N(B_3)=12$,
and $\varrho(d) \leq N(B_d)$ for every $d \geq 4$. In particular, $\varrho(d) = O(1.321^d)$.
\end{theorem}

A key fact for proving the theorem is the following easy lemma.

\begin{lemma} \label{lem:60}
  Let $B=B(o)$ be a unit ball centered at $o$, that is part of $\opt$.
  Let $p,q \in B$ be any two points in $B$ presented to the online algorithm that forced the
  algorithm to open new balls centered at $p$ and $q$; refer to Fig.~\ref{fig:60}.
  Then $\angle{poq} > \pi/3$.
  \end{lemma}

\begin{figure}[htbp]
\centering\includegraphics[scale=0.63]{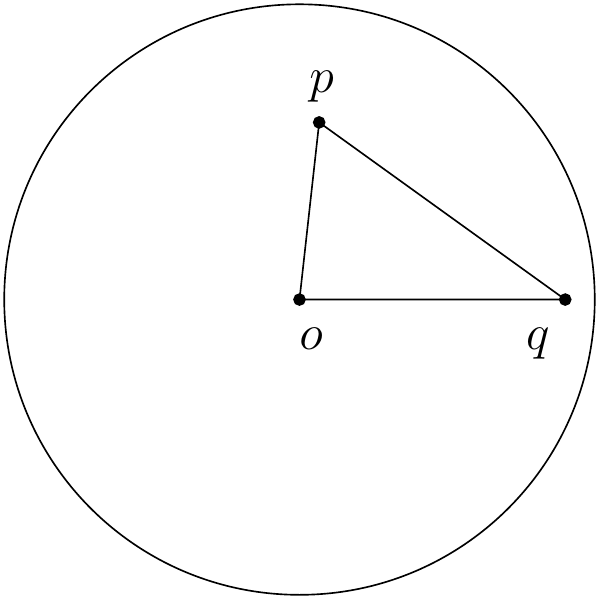}
\caption{Lemma~\ref{lem:60}.}
\label{fig:60}
\end{figure}

\begin{proof}
 Assume for contradiction that $\alpha= \angle{poq} \leq \pi/3$.
 Assume also, as we may, that $p$ arrives before $q$.
  Since $q \notin B(p)$, we have $|pq| > 1$.
  Consider the triangle $\Delta{poq}$; we may further assume that
  $\angle{opq} \geq \angle{oqp}$ (if we have the opposite inequality,
  the argument is symmetric). In particular, we have $\angle{opq} \geq \pi/3$.
Since $\angle{poq} \leq \pi/3$ and $\angle{opq} \geq \pi/3$, the law of sines yields that
$|oq| \geq |pq| > 1$. However, this contradicts the fact that $q$ is contained in $B$,
and the proof is complete.
\end{proof}

\begin{corollary} \label{cor:caps}
  Let $B=B(o)$ be a unit ball centered at $o$, that is part of $\opt$.
  Let $p \in B$ be any point in $B$ presented to the online algorithm that forced the
  algorithm to open a new ball centered at $p$; consider the cone $\Psi(p)$ with apex at $o$,
  axis $\vec{op}$, and angle $\pi/6$ around $\vec{op}$. Then the cones $\Psi(p)$ are pairwise
  disjoint in $B$; hence the corresponding caps on the surface of $B$ are also
  nonoverlapping.
\end{corollary}

\paragraph{Proof of Theorem~\ref{thm:centered}.}
For each unit ball $B$ of $\opt$ we bound from above the number of unit balls opened by
\texttt{Algorithm Centered} whose center lies in $B$. Suppose this number is at most $A$
(for each ball in $\opt$). Since the center of each unit ball opened by the algorithm is a point of the set
and all points in the set are covered by balls in $\opt$, it follows that the competitive ratio of
\texttt{Algorithm Centered} is at most $A$.

By Corollary~\ref{cor:caps} we are interested in the maximum number $A(\alpha)$ of
nonoverlapping caps $C(\alpha)$ that can be placed on $S_d$, for $\alpha=\pi/6$.
This is precisely the number of unit balls that can touch a given ball externally without
overlapping, namely the Newton number $N(B_d)$ in the same dimension~$d$.

For $d=2$ we gain $1$ in the bound due to the fact that the inequality in Lemma~\ref{lem:60}
is strict and we are dealing with the unit circle; the five vertices of a regular pentagon 
inscribed in a unit circle make a tight example with ratio $5$; note that the minimum pairwise distance
between points is $2 \sin (\pi/5) > 1$, and so the algorithm places a new ball for each point.
For $d=3$ the twelve vertices of a regular icosahedron inscribed in a unit sphere make a tight example
with ratio $12$; note that the minimum pairwise distance between points is $(\sin (2\pi/5))^{-1} > 1$,
and the same observation applies. 
\hfill$\qed$

\paragraph{Bounds on the Newton number of the ball.}
  A classic formula established by Rankin~\cite{Ra55} yields that
\begin{equation} \label{eq:rankin}
  N(B_d) \leq \frac{\pi^{1/2} \, \Gamma \left(  \frac{d-1}{2} \right)}
  {2 \sqrt2  \, \Gamma \left(  \frac{d}{2} \right) \,
    \int_0^{\pi/4} (\sin \theta)^{d-2} \, (\cos \theta -\cos \frac{\pi}{4}) \intd \theta} =A_d^*,
\end{equation}
and
\begin{equation} \label{eq:rankin2}
A_d^* \sim \sqrt{\frac{\pi}{8}}  \, d^{3/2} \, 2^{d/2}.
\end{equation}

More recently, Kabatiansky and Levenshtein~\cite{KL89} have established
a sharper upper bound
\begin{equation} \label{eq:kl}
  N(B_d) \leq 2^{0.401 d (1+o(1))}.
\end{equation}
In particular, $N(B_d) = O(1.321^d)$.
It is worth noting that the best lower known on the Newton number,
due to Jenssen~\etal~\cite{JJP18} is far apart; see also~\cite{Wy65}.
\begin{equation} \label{eq:jjp}
  N(B_d) = \Omega \left( d^{3/2} \cdot \left( \frac{2}{\sqrt3} \right)^d \right).
\end{equation}
In particular, $N(B_d) = \Omega(1.154^d)$.

\paragraph{Remark.} From an earlier discussion in Section~\ref{sec:intro},
it immediately follows that \texttt{Algorithm Centered} has a competitive ratio $O(1.321^d)$
for \textsc{Unit Clustering} in $\RR^d$ under the $L_2$ norm.
It is however worth noting that presently there is no online algorithm
for \textsc{Unit Clustering} in $\RR^d$ under the $L_\infty$ norm
with a competitive ratio $o(2^d)$. The best one known under this norm (for large $d$)
has ratio $2^d  \cdot \frac56$ for every $d \geq 2$; note that this is only marginally better than
the trivial $2^d$ ratio.

\section{Lower bounds on the competitive ratio for online unit covering in Euclidean $d$-space}
\label{sec:lower}

Theorem~\ref{thm:lower-d} that we prove in this section greatly improves
the previous best lower bound on the competitive ratio of a deterministic algorithm,
$\Omega(\log{d} / \log{\log{\log{d}}})$, due to Charikar~\etal~\cite{CCFM04}.

\paragraph{Previous lower bounds for $d=2,3$.}
To clarify matters, we briefly summarize the calculation leading to the previous best lower bounds
on the competitive ratio. Charikar~\etal~\cite{CCFM04} claim (on p.~1435) that a lower bound
of $4$ for $d=2$ under the $L_2$ norm follows from their Theorem~6.2;  but this claim appears unjustified;
only a lower bound of $3$ is implied. The proof uses a volume argument. For a given $d$,
the parameters $R_t$ are iteratively computed for $t=1,2,\ldots$ by using the recurrence relation
\begin{equation} \label{eq:volume}
  R_{t+1} = \frac{R_t + t^{1/d}}{2}, \text{      where } R_1=0.
\end{equation}
The lower bound on the competitive ratio of any deterministic algorithm given by the argument
is the largest $t$ for which $R_t \leq 1$. The values obtained for $R_t$, for $t=1,2,\ldots$
and $d=2,3$ are listed in Table~\ref{tab:volume}; as such, both lower bounds are equal to $3$. 
\begin {table} [htbp]
\centering
\begin{tabular}{|c|c|c|c|c|}
\hline
$d$ & $R_1$ & $R_2$ & $R_3$ & $R_4$ \\
\hline
\hline
$2$ & $0$ & $0.5$ & $0.957\ldots$ & $1.344\ldots$  \\ \hline
$3$ & $0$ & $0.5$ & $0.879\ldots$ & $1.161\ldots$  \\ \hline
\end{tabular}
\caption {Values $R_t$, for $t=1,2,\ldots$}
\label {tab:volume}
\end {table}

\subsection{A new lower bound in the plane}

In this section, we deduce an improved lower bound of $4$
(an alternative proof will be provided by Theorem~\ref{thm:lower-d}). 

\begin{theorem}\label{thm:lower-2}
  The competitive ratio of any deterministic online algorithm for \textsc{Unit Covering}
  in the plane (in the $L_2$ norm) is at least $4$.
\end{theorem}
\begin{proof} Consider a deterministic online algorithm $\alg$.
We present an input instance $\sigma$ for $\alg$
and show that the solution $\alg(\sigma)$ is at least $4$ times $\opt(\sigma)$.
Our proof works like a two player game, played by Alice and Bob. Here,
Alice is presenting  points to Bob, one at a time. Bob (who plays the role of the algorithm)
takes the decision whether to place a new disk or not. If a new disk is
required, Bob decides where to place it. Alice tries to force Bob to
place as many new disks as possible by presenting the points in a smart way.
Bob tries to place new disks in a way such that they
may cover other points presented by Alice in the future, thereby
reducing the need of placing new disks quite often.
\begin{figure}[htpb]
\centering\includegraphics[scale=0.75]{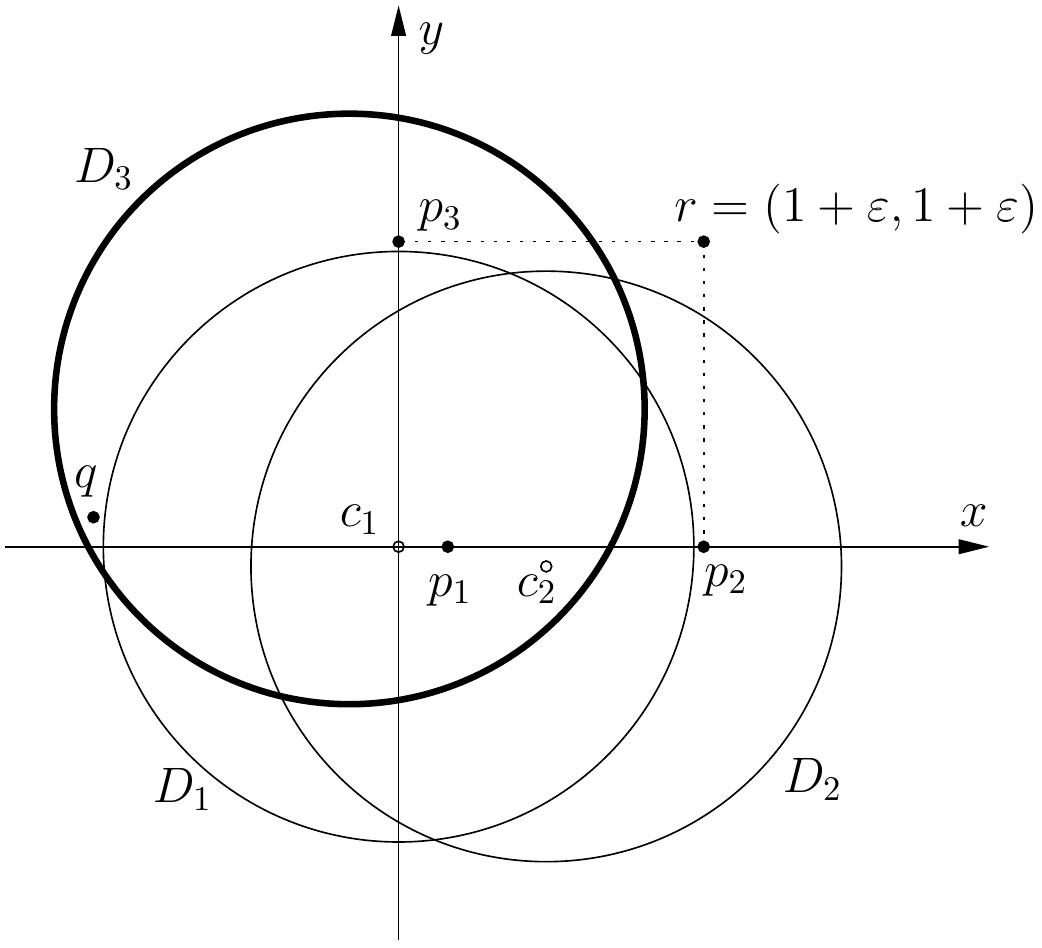}
\caption{A lower bound of $4$ on the competitive ratio in the plane.
The figure illustrates the case $p_4=r$.}
\label{fig:f21}
\end{figure}		

The center of a disk $D_i$ is denoted by $c_i$, $i=1,2,\ldots$.
The point coordinates will depend on a parameter $\eps>0$; a sufficiently small
$\eps \leq 0.01$ is chosen so that the inequalities appearing in the proof hold.
First, point $p_1$ arrives and the algorithm places disk $D_1$ to cover it. By symmetry,
it can be assumed that $c_1=(0,0)$ and  $p_1=(x,0)$, where $0 \leq x \leq 1$.
The second point presented is $p_2=(1+\eps^2,0)$ and,
since $p_2 \notin D_1$, a second disk $D_2$ is placed to cover it.
By symmetry, it can be assumed that $y(c_2) \leq 0$.
The third point presented is $p_3=(0,1+\eps)$, and neither $D_1$ nor $D_2$ covers it;
thus a new disk, $D_3$, is placed to cover $p_3$.

Consider two other candidate points, $q=(-1+\eps, \sqrt{2 \eps})$ and $r=(1+\eps,1+\eps)$.
Since
$$ |q c_1|^2 = (-1 +\eps)^2 + 2 \eps = 1+ \eps^2 - 2 \eps + 2 \eps = 1 + \eps^2 >1, $$
$q$ is not covered by $D_1$;  and clearly $r$ is not covered by $D_1$.
Since
\begin{align*}
  |q c_2|^2 &\geq (1 -\eps + \eps^2)^2  + 2 \eps
= 1+ \eps^2 + \eps^4 - 2 \eps + 2 \eps^2 -2 \eps^3 + 2 \eps\\
&= 1 + 3 \eps^2 +O(\eps^3) > 1,
\end{align*}
$q$ is not covered by $D_2$; and clearly $r$ is not covered by $D_2$.
Note also that the $D_3$ cannot cover both $q$ and $r$, since their distance is close to $\sqrt5>2$.
We now specify $p_4$, the fourth point presented to the algorithm.
If $q$ is covered by $D_3$, let $p_4=r$, otherwise let $p_4=q$.
In either case, a fourth disk, $D_4$, is required to cover $p_4$.

To conclude the proof, we verify that $p_1,p_2,p_3,p_4$ can be covered by a unit disk.

\medskip
\emph{Case 1: $p_4=r$.} It is easily seen that $p_1,p_2,p_3, p_4$ can be covered by the
unit disk $D$ centered at $\left(\frac12,\frac12\right)$; indeed, the four points are close to
the boundary of the unit square $[0,1]^2$.

\medskip
\emph{Case 2: $p_4=q$.} Consider the unit disk $D$ centered at the midpoint $c$ of $q p_2$. We have
\begin{align*}
  |q p_2|^2 &= (2 -\eps + \eps^2)^2 + 2 \eps
  = 4 + \eps^2 + \eps^4 - 4 \eps + 4 \eps^2 -2 \eps^3 + 2 \eps \\
  &= 4 - 2 \eps + O(\eps^2) < 4.
\end{align*}
It follows that $D$ covers $p_2$ and $p_4$.
Note that
$$ c =\left( \frac{\eps+\eps^2}{2}, \sqrt{\frac{\eps}{2}} \right). $$
We next check the containment of $p_1$ and $p_3$.
$$ |c p_1|^2 \leq \left(1 -\frac{\eps + \eps^2}{2} \right)^2 + \frac{\eps}{2}
= 1 - \frac{\eps}{2} +O(\eps^2) <1, $$
thus $D$ also covers $p_1$.
Finally, we have
\begin{align*}
|c p_3|^2 &\leq \left(\frac{\eps + \eps^2}{2} \right)^2 + \left(1 + \eps - \sqrt{\frac{\eps}{2}} \right)^2
 \leq \eps^2 + \left(1 + \eps - \sqrt{\frac{\eps}{2}} \right)^2\\
 &= \eps^2 + 1 + \eps^2 + \frac{\eps}{2} + 2 \eps - \sqrt{ 2 \eps} -  \sqrt{ 2 \eps^3} \\
 &= 1 - \sqrt{ 2 \eps} +O(\eps) <1,
\end{align*}
thus $D$ also covers $p_3$.

\medskip
We have shown that $\alg(\sigma)/\opt(\sigma) \geq 4$, and the proof is complete.
\end{proof}

\subsection{A new lower bound in $d$-space}

We introduce some additional terminology.
For every integer $k$, $0\leq k<d$, a \emph{$k$-sphere} of radius $r$
centered at a point $c\in\RR^d$
is the locus of points in $\RR^d$ at distance $r$ from a center $c$,
and lying in a $(k+1)$-dimensional affine subspace that contains $c$.
In particular, a $(d-1)$-sphere of radius $r$ centered at $c$ is
the set of \emph{all} points $p\in \mathbb{R}^d$ such that $|cp|=r$;
a 1-sphere is a circle lying in a $2$-dimensional affine plane;
and a 0-sphere is a pair of points whose midpoint is $c$.
A~\emph{$k$-hemisphere} is a $k$-dimensional manifold with boundary,
defined as the intersection $S\cap H$, where $S$ is a $k$-sphere
centered at some point $c\in \mathbb{R}^d$ and $H$ is a halfspace
whose boundary $\partial H$ contains $c$ but does not contain~$S$.
For $k\geq 1$, the \emph{relative boundary} of the $k$-hemisphere $S \cap H$
is the $(k-1)$-sphere $S \cap (\partial H)$ concentric with $S$;
and the \emph{pole} of $S\cap H$ is the unique point $p\in H$ such that  $\overrightarrow{cp}$
is orthogonal to the $k$-dimensional affine subspace that contains $S\cap (\partial H)$.
For $k=0$, a 0-hemisphere consists of a single point, and we define
the \emph{pole} to be that point.
We make use of the following observation.

\begin{observation} \label{obs:balls}
Let $S$ be a $k$-sphere of radius $1+\varepsilon$, where $0\leq k<d$
and $\varepsilon>0$; and let $B$ be a unit ball in $\mathbb{R}^d$.
Then $S\setminus B$ contains a $k$-hemisphere.
\end{observation}
\begin{proof}
Without loss of generality, $S$ is centered at the origin, and lies in
the subspace spanned by the coordinate axes $x_1,\ldots, x_{k+1}$. By
symmetry, we may also assume that the center of $B$ is on the
nonnegative $x_1$-axis, say, at $(b,0,\ldots ,0)$ for some $b\geq
0$. If $b=0$, then $S$ and $B$ are concentric and $B$ lies in the
interior of $S$, consequently, $S\setminus B=S$. Otherwise, $S\cap B$
lies in the open halfspace $x_1>0$, and $S\setminus B$ contains
the $k$-hemisphere $S\cap \{(x_1,\ldots , x_d)\in \mathbb{R}^d: x_1\leq 0\}$.
\end{proof}

\begin{theorem} \label{thm:lower-d}
The competitive ratio of every deterministic online algorithm
(with an adaptive deterministic adversary) for \textsc{Unit Covering}
in $\RR^d$ under the $L_{2}$ norm is at least $d+1$ for every $d \geq 1$;
and at least $d+2$ for $d =2,3$.
\end{theorem}
\begin{proof}
Consider a deterministic online algorithm $\alg$.
We present an input instance $\sigma$ for $\alg$ and show that the solution
$\alg(\sigma)$ is at least $d+1$ times $\opt(\sigma)$.
In particular, $\sigma$ consists of $d+1$ points in $\mathbb{R}^d$ that fit in a unit ball,
hence $\opt(\sigma)=1$, and we show that $\alg$ is required to place
a new unit ball for each point in $\sigma$.
Similarly to the proof of Theorem~\ref{thm:lower-2},
our proof works like a two player game between Alice and Bob.

Let the first point $p_0=o$ be the origin in $\mathbb{R}^d$ (we will use either notation as
convenient). 
For a constant $\varepsilon \in (0,\frac{1}{2d})$, let $S_0$ be the $(d-1)$-sphere
of radius $1+\varepsilon$ centered at the origin $o$. Refer to Fig.~\ref{fig:lower-d}.
Next, $B_0$ is placed to cover $p_0$. 
The remaining points $p_1,\ldots , p_d$ in $\sigma$ are chosen adaptively,
depending on Bob's moves. We maintain the following two invariants:
For $i=1,\ldots , d$, when Alice has placed points $p_0,\ldots, p_{i-1}$,
and Bob placed unit balls $B_0,\ldots B_{i-1}$,
\begin{enumerate}[(I)] \itemsep 2pt
\item\label{inv:1} the vectors $\overrightarrow{op_j}$, for $j=1,\ldots , i-1$,
    are pairwise orthogonal and have length $1+\varepsilon$;
\item\label{inv:2} there exists a $(d-i)$-hemisphere $H_i\subset S_0$
    that lies in the $(d-i+1)$-dimensional subspace orthogonal
    to $\langle \overrightarrow{op_j}: j=1,\ldots , i-1\rangle$
    and is disjoint from $\bigcup_{j=0}^{i-1}B_j$.
\end{enumerate}
Both invariants hold for $i=1$: \eqref{inv:1} is vacuously true, and \eqref{inv:2}
holds by Observation~\ref{obs:balls} (the first condition of \eqref{inv:2} is vacuous in this case). 
\begin{figure}[htbp]
\centering\includegraphics{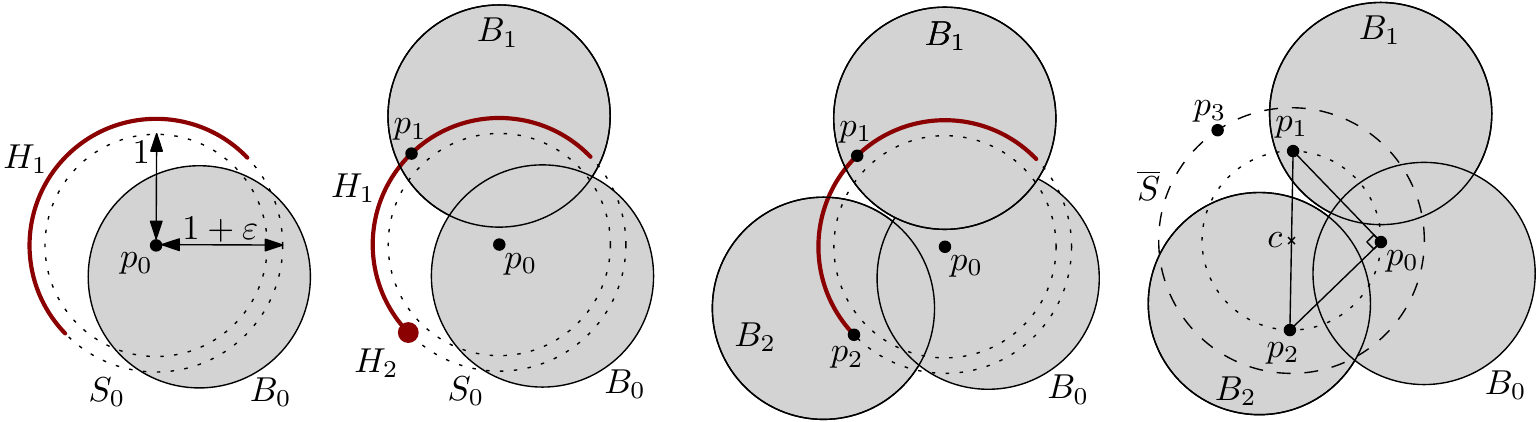}
\caption{The first three steps of the game between Alice and Bob in
  the proof of Theorem~\ref{thm:lower-d} for $d=2$.
After the 3rd step, Alice can place a 4th point $p_3\in \overline{S}$
which is not covered by the balls $B_0,B_1,B_2$.}
\label{fig:lower-d}
\end{figure}

At the beginning of step $i$ (for $i=1,\ldots, d$), assume that both invariants hold.
Alice chooses $p_i$ to be the pole of the $(d-i)$-hemisphere $H_i$.
By Invariant~\eqref{inv:2}, $p_i$ is not covered by $B_0,\ldots , B_{i-1}$,
and Bob has to choose a new unit ball $B_i$ that contains $p_i$.
By Invariant~\eqref{inv:1}, $H_i\subset S_0$, so $|op_i|=1+\varepsilon$.
By Invariant~\eqref{inv:2}, $\overrightarrow{op_i}$ is orthogonal
to the vectors $\overrightarrow{op_j}$, for $j=1,\ldots , i-1$.
Hence Invariant~\eqref{inv:1} is maintained.

Let $S_i$ be the relative boundary of $H_i$, which is a
$(d-i-1)$-sphere centered at the origin.
Since $p_i$ is the pole of $H_i$, $\overrightarrow{op_i}$
is orthogonal to the $(d-i)$-dimensional subspace spanned by $S_i$.
By Observation~\ref{obs:balls}, $S_i$ contains
a $(d-i-1)$-hemisphere that is disjoint from $B_i$.
Denote such a $(d-i-1)$-hemisphere by $H_{i+1}\subset S_i$.
Clearly, $H_{i+1}$ is disjoint from the balls
$B_0,\ldots, B_{i-1}, B_i$; so Invariant~\eqref{inv:2}
is also maintained.

By construction, $p_i$ ($i=1,\ldots, d$) is not covered by the balls
$B_0,\ldots, B_{i-1}$, so Bob has to place a unit ball for each of the
$d+1$ points $p_0,p_1,\ldots,p_d$.
By Invariant~\eqref{inv:1}, the points $p_1,\ldots , p_d$ span a regular
$(d-1)$-dimensional simplex of side length $(1+\varepsilon)\sqrt{2}$.
By Jung's Theorem~\cite[p.~46]{HB64}, the radius of the smallest enclosing
ball of $p_1,\ldots , p_d$ is
$$ R = (1+\varepsilon)\sqrt{2}\cdot\sqrt{\frac{d-1}{2d}}
<\left(1+\frac{1}{2d}\right)\sqrt{\frac{d-1}{d}}
=\sqrt{\frac{(2d+1)^2(d-1)}{4d^3}}
=\sqrt{\frac{4d^3-3d-1}{4d^3}}<1, $$
and this ball contains the origin $p_0$, as well.

We next show how to adjust the argument to derive a slightly better lower bound
of $d+2$ for $d=2,3$. 
Let $B$ be the smallest enclosing ball of the points $p_0,p_1,\ldots,p_d$,
and let $c$ be the center of $B$. As noted above, the radius of $B$ is
$R=(1+\varepsilon)\sqrt{(d-1)/d}$. Let $\overline {S}$ be the $(d-1)$-sphere of radius
$2-R = 2-(1+\varepsilon)\sqrt{(d-1)/d}$ centered at $c$. Then
the smallest enclosing ball of $B$ and an arbitrary point $p_{d+1}\in \overline{S}$
has unit radius. That is, points $p_0,\ldots , p_d,p_{d+1}$ fit in a unit ball.
This raises the question whether Alice can choose yet another point $p_{d+1}\in \overline{S}$
outside of the balls $B_0,\ldots, B_d$ placed by Bob.

For $d=2$, $\overline{S}$ has radius
$2-(1+\varepsilon)\sqrt{1/2}=2-(1+\varepsilon)(\sqrt{2}/2) \geq 1.2928$
(provided that $\eps>0$ is sufficiently small). 
A unit disk can cover a circular arc in $\overline{S}$ of diameter at most $2$.
If $3$ unit disks can cover $\overline{S}$, then $\overline{S}$ would be the smallest enclosing
circle of a triangle of diameter at most $2$, and its radius would be
at most $\frac{2}{3}\sqrt{3} \leq 1.1548$ by Jung's Theorem.
Consequently, Alice can place a 4th point $p_3 \in \overline{S}$ outside of $B_0,B_1,B_2$,
and all four points $p_0,\ldots,p_3$ fit in a unit disk; see Fig.~\ref{fig:lower-d}\,(right)
for an example. That is, $\alg(\sigma)=4$ and $\opt(\sigma)=1$; and we
thereby obtain an alternative proof of Theorem~\ref{thm:lower-2}.

For $d=3$, $\overline{S}$ has radius $R_1= 2-(1+\varepsilon) (\sqrt{2/3}) \geq 1.1835$
(provided that $\eps>0$ is sufficiently small). 
Let $c_i$ denote the center of $B_i$, for $i=0,1,2,3$;
we may assume that at least one of the balls $B_i$, say $B_0$, is not concentric with $\overline{S}$,
since otherwise $\bigcup_{i=0}^{3}B_i$ would cover zero area of $\overline{S}$.
We may also assume for concreteness that $c c_0$ is a vertical segment;
let $\pi_0$ denote the horizontal plane incident to $c$. 
Then $C = \overline{S} \cap \pi_0$ is horizontal great circle (of radius $R_1$) centered at $c$.
Note that $C \cap B_0 =\emptyset$, and so if $\bigcup_{i=0}^{3} B_i$ covers $\overline{S}$,
then $\bigcup_{i=1}^{3} (B_i \cap \pi_0)$ covers $C$. However, the analysis of the planar case ($d=2$)
shows that this is impossible; indeed, we have $R_1 \geq 1.1835 > 1.1548$.
Consequently,  Alice can place a 5th point $p_4 \in \overline{S}$ outside of $B_0,B_1,B_2,B_3$,
and all five points $p_0,\ldots,p_4$ fit in a unit ball.
That is, $\alg(\sigma)=5$ and $\opt(\sigma)=1$ and a lower bound of $5$ 
on the competitive ratio is implied.
\end{proof}


\section{Unit covering for lattice points in the plane} \label{sec:lattice}

In this section, we give optimal deterministic algorithms for online \textsc{Unit Covering}
of points from the infinite unit square and hexagonal lattices.
The input points always belong to the lattice under consideration.
We start with the infinite unit square lattice $\ZZ^2$.

\begin{theorem} \label{thm:lattice1}
  There exists a deterministic online algorithm for online \textsc{Unit Covering} of integer points
  (points in $\ZZ^2$) with competitive ratio $3$. This result is tight:  the competitive ratio of
  any deterministic online algorithm for this problem is at least $3$.
\end{theorem}
\begin{proof} First, we prove the lower bound; refer to Fig.~\ref{fig:f20-f15}\,(left).
  First, point $p_1$ arrives and disk $D_1$ covers it. Observe that $D_1$ misses at least one point
  from $\{p_2,p_3\}$, since $|p_2 p_3|=2 \sqrt2>2$. We may assume that $D_1$ missed $p_2$;
  and this further implies that $D_1$ does not cover $p_4$ or $p_5$, since if it would,
  it would also cover $p_2$, a contradiction. Now, $D_2$ is placed to cover $p_2$.
\begin{figure}[htbp]
\begin{center}
  \includegraphics[scale=0.85]{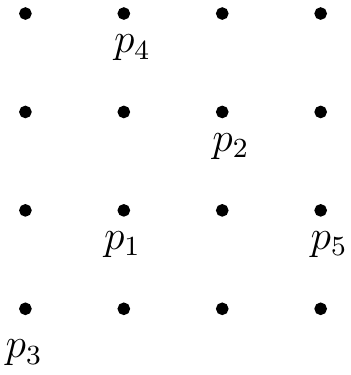}
  \hspace{12mm}
  \includegraphics[scale=0.41]{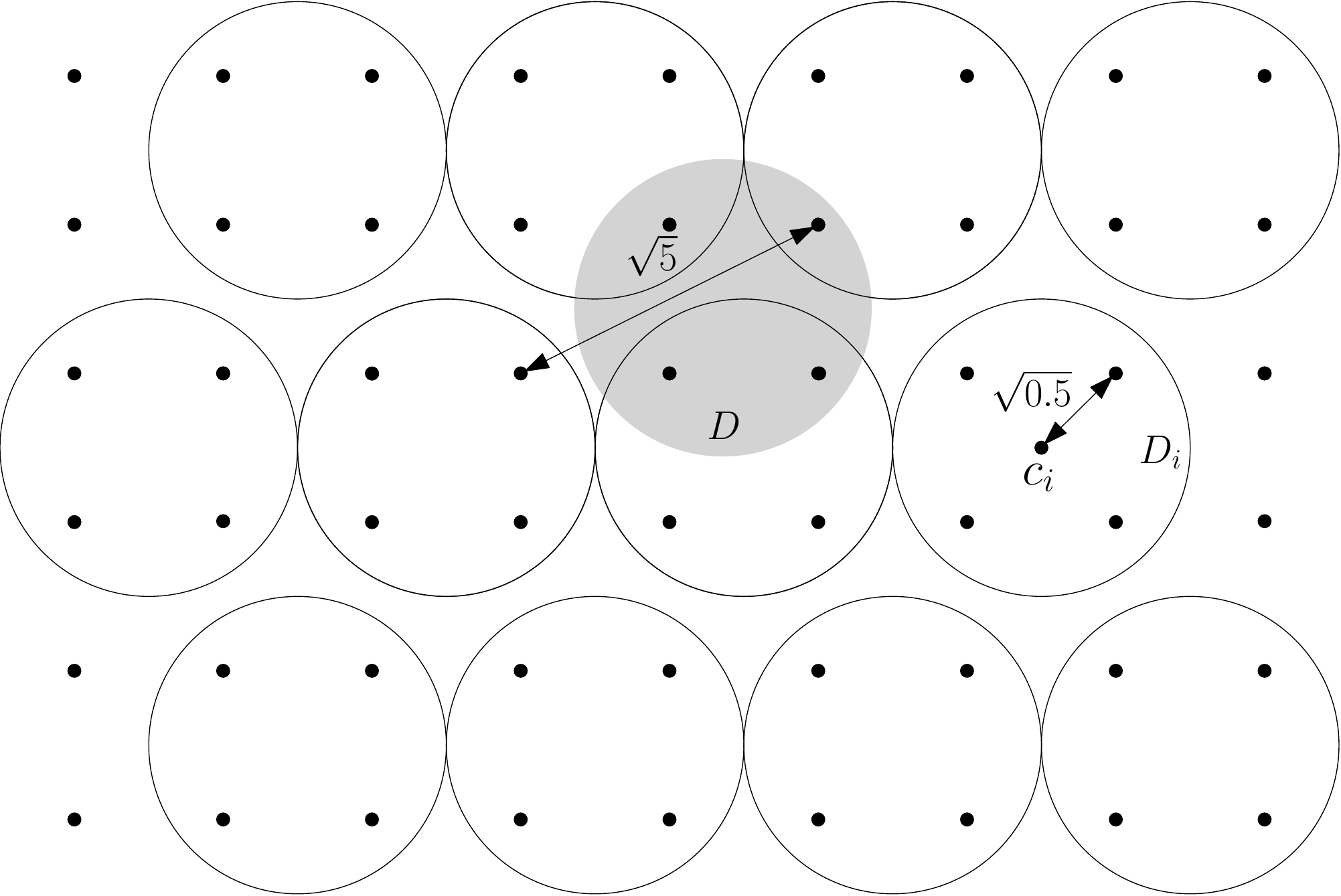}
\end{center}
\caption{Left: Lower bound for $\ZZ^2$.
Right: Illustration of the upper bound; the disk $D$ is shaded.}
\label{fig:f20-f15}
\end{figure}
If $D_2$ covers $p_4$, next input point is $p_5$, otherwise it is $p_4$.
In either case a third disk is needed.
To finish the proof, observe that $\{p_1,p_2,p_4\}$ or $\{p_1,p_2,p_5\}$ can be covered
optimally by a single unit disk; hence the competitive ratio of any deterministic algorithm
is at least~$3$.

Next, we present an algorithm which has competitive ratio~$3$. Refer to Fig.~\ref{fig:f20-f15}\,(right).
Partition the lattice points using unit disks as shown in the figure. When a point arrives,
use the disk it belongs to in the partition.
For the analysis, consider a disk $D$ from an optimal cover. As seen in the figure,
$D$ can cover points which belong to at most three disks used for partitioning the lattice.
Thus, we conclude that the algorithm has competitive ratio $3$.
\end{proof}

In the following, we state our result for the infinite hexagonal lattice.

\begin{theorem}\label{thm:lattice2}
  There exists a deterministic online algorithm for online \textsc{Unit Covering} of
  points of the hexagonal lattice with competitive ratio $3$. This result is tight:
  the competitive ratio of any deterministic online algorithm for this problem is at least $3$.
\end{theorem}
\begin{proof} The proof  is similar to that of Theorem~\ref{thm:lattice1}.
  We start by proving the lower bound of $3$; refer to Fig.~\ref{fig:f22-f18}\,(left).
\begin{figure}[htbp]
  \begin{center}
    \includegraphics[scale=0.60]{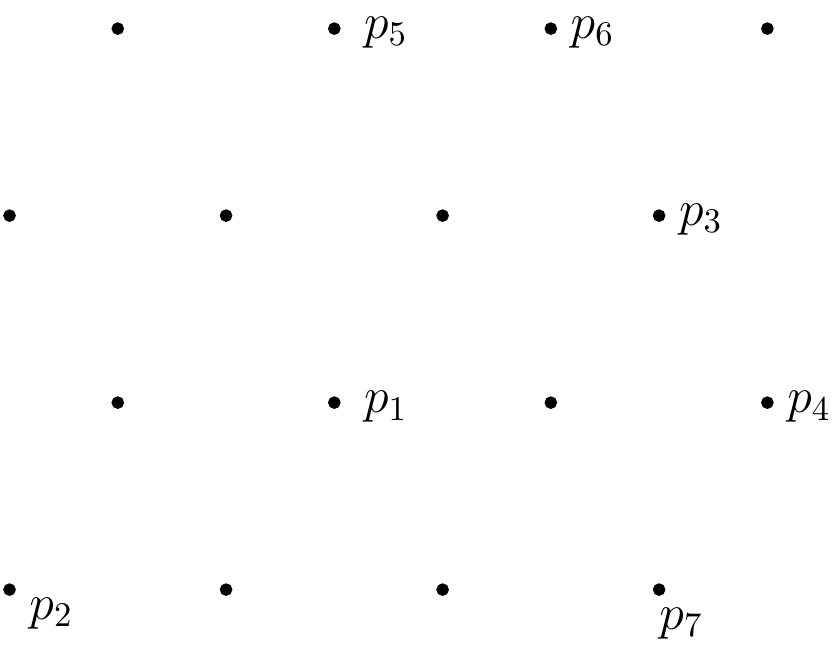}
    \hspace{12mm}
    \includegraphics[scale=0.7]{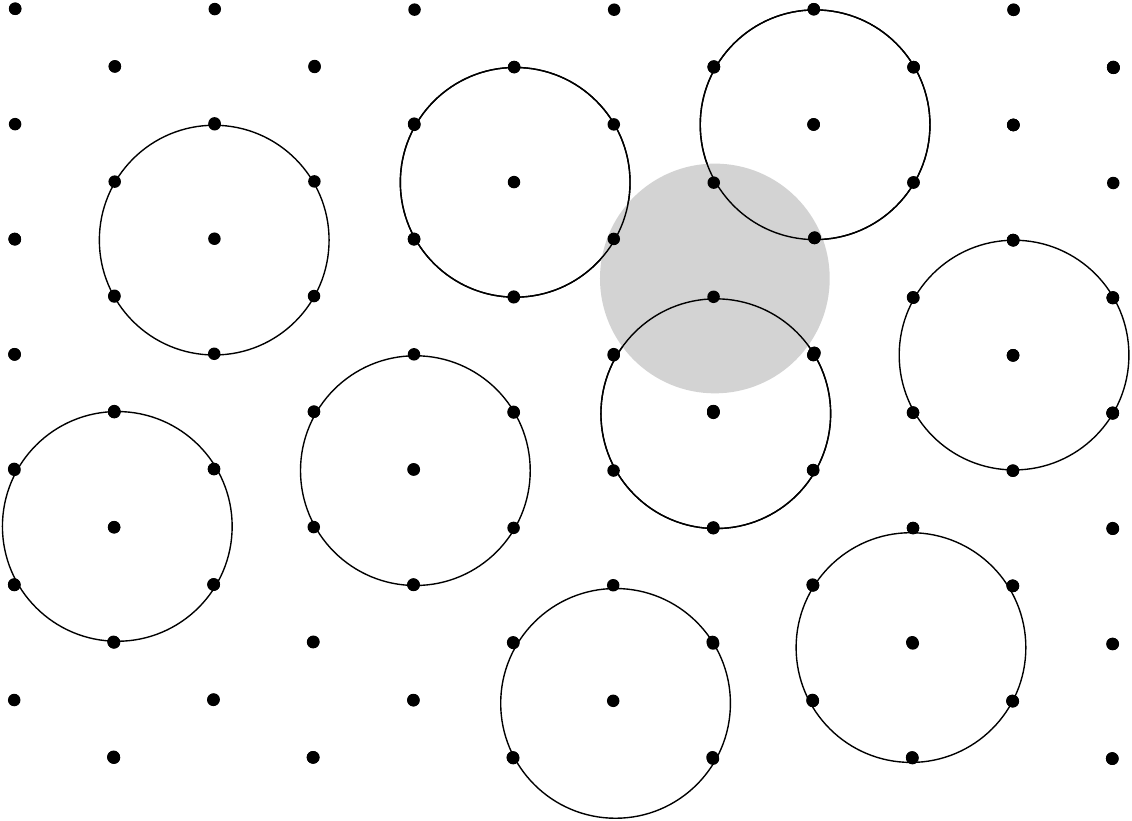}
\end{center}
\caption{Left: Lower bound for the hexagonal lattice.
Right: Illustration of the upper bound.}
\label{fig:f22-f18}
\end{figure}
The first point, $p_1$, arrives and $D_1$ is used to cover it.
$D_1$ misses at least one of $\{p_2,p_3\}$, since $|p_2 p_3|= 2\sqrt3 >2$.
By symmetry, we may assume that $D_1$ misses $p_3$.
The algorithm uses $D_2$ to cover $p_3$.
We distinguish two cases:

\medskip
\emph{Case 1: $D_2$ misses $p_4$.} Since $D_1$ misses $p_3$, $D_1$ also misses $p_4$.
Otherwise, if $D_1$ covers $p_4$, then $D_1$ also covers $p_3$, a contradiction.
The algorithm uses $D_3$ to cover $p_4$. Thus the ratio is $3$ since $p_1,p_3,p_4$ can be covered
by a single disk centered at $(p_1+p_4)/2$, and the algorithm has used three disks, $D_1,D_2,D_3$.

\medskip
\emph{Case 2: $D_2$ covers $p_4$.} This means that $D_2$ misses $p_5$.
Now, two things may happen:

\begin{enumerate} \itemsep 0pt
\item \emph{$D_1$ misses $p_5$ too.} Then $p_5$ is the next input point, and the
algorithm uses $D_3$ to cover it. Here $p_1,p_3,p_5$ can be covered by a single disk
centered at $(p_1+p_3+p_5)/3$, but the algorithm has used three disks, $D_1,D_2,D_3$.
\item \emph{$D_1$ covers $p_5$.} Since $D_1$ does not cover $p_3$,
  $D_1$ cannot cover $p_6$. If $D_2$ misses $p_6$, let $p_6$ be the third point presented;
  the algorithm uses $D_3$ to cover $p_6$. Here $p_1,p_3,p_6$ can be covered by a single disk
centered at $(p_1+p_6)/2$, but the algorithm has used three disks, $D_1,D_2,D_3$.
If $D_2$ covers $p_6$, let $p_7$ be the third point presented.
Note that $D_1$ cannot cover $p_7$ since it covers $p_5$; also,
$D_2$ cannot cover $p_7$ since it covers $p_6$.
The algorithm uses $D_3$ to cover $p_7$. Here $p_1,p_3,p_7$ can be covered by a single disk
centered at $(p_1+p_3+p_7)/3$, but the algorithm has used three disks, $D_1,D_2,D_3$.
\end{enumerate}
In all cases a lower bound of $3$ has been enforced by Alice, as required.

Now, we prove the upper bound of $3$. As in the case of the unit square lattice,
we partition the lattice points using disks as shown in Fig.~\ref{fig:f22-f18}\,(right).
Arguing similarly, it can be concluded that the same algorithm has competitive ratio $3$
in this case.
\end{proof}

\section{Conclusion} \label{sec:conclusion}

Our results suggest several directions for future study.
For instance, the gap between the sublinear lower bound and the exponential upper bound in
the competitive ratios for online \textsc{Unit Covering} is intriguing.
We summarize a few specific questions of interest.

\begin{problem}
  Is there a lower bound on the competitive ratio for \textsc{Unit Covering}
  that is exponential in $d$? Is there a superlinear lower bound?
\end{problem}

\begin{problem}
Can the online algorithm for integer points (with ratio $3$ in the plane)
be extended to higher dimensions, \ie, for covering points in $\ZZ^d$?
What ratio can be obtained for this variant?
\end{problem}


\begin{thebibliography}{99}

\bibitem{AAA+09}
Noga Alon, Baruch Awerbuch, Yossi Azar, Niv Buchbinder, and Joseph Naor,
The online set cover problem,
\emph{SIAM J. Comput.} \textbf{39(2)} (2009), 361--370.

\bibitem{BLMS17}
Ahmad Biniaz, Peter Liu, Anil Maheshwari, and Michiel Smid,
Approximation algorithms for the unit disk cover problem in 2D and 3D,
\emph{Comput. Geom.} \textbf{60} (2017), 8--18.

\bibitem{BY98}
Allan Borodin and Ran El-Yaniv,
\emph{Online Computation and Competitive Analysis},
Cambridge University Press, Cambridge, 1998.

\bibitem{BDM15}
Peter Boyvalenkov, Stefan Dodunekov, and Oleg R. Musin,
A survey on the kissing numbers,
\emph{Serdica Math. J.}
\textbf{38} (2012), 507--522.
Preprint available at \url{arXiv.org/abs/1507.03631}.

\bibitem{BMP05}
Peter Brass, William Moser, and J\'anos Pach,
\emph{Research Problems in Discrete Geometry}, Springer, New York, 2005.

\bibitem{BN09b}
Niv Buchbinder and Joseph Naor,
Online primal-dual algorithms for covering and packing,
\emph{Math. Oper. Res.} \textbf{34(2)} (2009), 270--286.

\bibitem{CZ09}
Timothy M. Chan and Hamid Zarrabi-Zadeh,
A randomized algorithm for online unit clustering,
\emph{Theory Comput. Syst.} \textbf{45(3)} (2009), 486--496.

\bibitem{CCFM04}
Moses Charikar, Chandra Chekuri, Tom\'as Feder, and Rajeev Motwani,
Incremental clustering and dynamic information retrieval,
\emph{SIAM J. Comput.} \textbf{33(6)} (2004), 1417--1440.


\bibitem{DT17}
Adrian Dumitrescu and Csaba D. T\'oth,
Online unit clustering in higher dimensions,
\emph{Proc. 15th International Workshop on Approximation and Online Algorithms (WAOA)},
LNCS~10787, Springer, Cham, 2017, pp.~238--252.

\bibitem{ERS98}
  Yves Edel, Eric M. Rains, and Neil J. A. Sloane,
 On kissing numbers in dimensions $32$ to $128$,
\emph{Electron. J. Combin.} \textbf{5} (1998), \#R22.

\bibitem{EL13}
Martin R. Ehmsen and Kim S. Larsen,
Better bounds on online unit clustering,
\emph{Theoret. Comput. Sci.} \textbf{500} (2013), 1--24.

\bibitem{ES10}
Leah Epstein and Rob van Stee,
On the online unit clustering problem,
\emph{ACM Trans. Algorithms} \textbf{7(1)} (2010), 1--18.

\bibitem{FG88}
Tom\'as Feder and Daniel H. Greene,
Optimal algorithms for approximate clustering,
in {\em Proc. 20th Annual ACM Symposium on Theory of Computing (STOC)},
1988, pp. 434--444.

\bibitem{FPT81}
Robert J. Fowler, Mike Paterson, and Steven L. Tanimoto,
Optimal packing and covering in the plane are NP-complete,
\emph{Inform. Process. Lett.} \textbf{12(3)} (1981), 133--137.

\bibitem{Go85}
Teofilo F. Gonzalez,
Clustering to minimize the maximum intercluster distance,
\emph{Theoret. Comput. Sci.} \textbf{38} (1985), 293--306.


\bibitem{HB64} Hugo Hadwiger and Hans Debrunner,
  \emph{Combinatorial Geometry in the Plane}
(English translation by Victor Klee), Holt, Rinehart and Winston, New York, 1964.

\bibitem{HM85}
Dorit S. Hochbaum and Wolfgang Maass,
Approximation schemes for covering and packing problems in image processing and VLSI,
\emph{J. ACM} \textbf{32(1)} (1985), 130--136.

\bibitem{JJP18}
Matthew Jenssen, Felix Joos, and Will Perkins,
On kissing numbers and spherical codes in high dimensions,
\emph{Advances in Mathematics} \textbf{335} (2018), 307-–321.

\bibitem{KL89}
Grigory A. Kabatiansky and Vladimir I. Levenshtein,
Bounds for packings on a sphere and in space,
\emph{Probl. Inform. Transm.} \textbf{14(1)} (1989), 1--17.

\bibitem{KK15}
Jun Kawahara and Koji M. Kobayashi,
An improved lower bound for one-dimensional online unit clustering,
\emph{Theoret. Comput. Sci.} \textbf{600} (2015), 171--173.

\bibitem{LH10}
Chen Liao and Shiyan Hu,
Polynomial time approximation schemes for minimum disk cover problems,
\emph{J. Comb. Optim.} \textbf{20} (2010), 399--412.


\bibitem{MS84}
Nimrod Megiddo and Kenneth J. Supowit,
On the complexity of some common geometric location problems,
\emph{SIAM J. Comput.} \textbf{13(1)} (1984), 182--196.

\bibitem{MR10}
Nabil H. Mustafa and Saurabh Ray,
Improved results on geometric hitting set problems,
\emph{Discrete Comput. Geom.} \textbf{44} (2010), 883--895.

\bibitem{Ra55}
Robert A. Rankin,
The closest packing of spherical caps in $n$ dimensions,
\emph{Glasgow Mathematical Journal} \textbf{2(3)} (1955), 139--144.

\bibitem {Va01} Vijay Vazirani,
\emph{Approximation Algorithms},
Springer Verlag, New York, 2001.

\bibitem {WS11} David P. Williamson and David B. Shmoys,
\emph{The Design of Approximation Algorithms},
Cambridge University Press, Cambridge, 2011.

\bibitem {Wy65} Aaron D. Wyner,
Capabilities of bounded discrepancy decoding,
\emph{Bell Sys. Tech. J.} \textbf{44} (1965), 1061--1122.

\bibitem{EdWynn}
Ed Wynn,
Covering a unit ball with balls half the radius,
\url{https://mathoverflow.net/questions/98007/covering-a-unit-ball-with-balls-half-the-radius}.

\bibitem{ZC09}
Hamid Zarrabi-Zadeh and Timothy M. Chan,
An improved algorithm for online unit clustering,
\emph{Algorithmica} \textbf{54(4)} (2009), 490--500.

\end{thebibliography}
\end{document}